\documentclass[10pt]{amsart}
\usepackage{amssymb}
\usepackage{amsthm}

\usepackage{amssymb,amsmath,amsthm,amsxtra,amsfonts}

\usepackage{xcolor}
\usepackage{colortbl}
\usepackage{graphicx}
\usepackage{multirow}

\newcommand{\bbC}{{\mathbb{C}}}
\newcommand{\bbD}{{\mathbb{D}}}

\newcommand{\bbR}{{\mathbb{R}}}

\newcommand{\bbZ}{{\mathbb{Z}}}

\newcommand{\calH}{{\mathcal H}}
\newcommand{\calJ}{{\mathcal J}}

\newcommand{\calT}{{\mathcal T}}

\newcommand{\beq}{\begin{equation}}
\newcommand{\eeq}{\end{equation}}
\newcommand{\ba}{\begin{align*}}
\newcommand{\ea}{\end{align*}}

\allowdisplaybreaks
\numberwithin{equation}{section}

\newtheorem{theorem}{Theorem}
\newtheorem{lemma}{Lemma}

\newtheorem*{definition}{Definition}
\theoremstyle{remark}
\newtheorem*{remarks}{Remarks}

\begin{document}

%

\title{On Gaussian random matrices coupled to the discrete Laplacian}
\author{Rostyslav Kozhan}
\address{Uppsala University; Department of Mathematics; Uppsala, Sweden}

\begin{abstract}
We study operators obtained by coupling an $n\times n$ random matrix from one of the Gaussian ensembles to the discrete Laplacian. We find the joint distribution of the eigenvalues and resonances of such operators. 
This is one of the possible mathematical models for quantum scattering in a complex physical system 
with one semi-infinite lead attached.
\end{abstract}

\maketitle


\section{Introduction}

Given a random Hermitian $n\times n$ matrix $\calH$ from one of the classical Gaussian ensembles,
we consider the operator on $\ell^2(\bbZ_+)$ obtained from $\calH$ by coupling it to the discrete Laplacian as follows:
\begin{equation}\label{eq:main_operator}
{\scalebox{1.3}{$\widetilde{\calH}=$}}
\left(
\begin{array}
{  c r | c  c  c }
\multicolumn{2}{c|}{{\cellcolor{black!10}}} & & & \phantom{\ddots}\\
\multicolumn{2}{c|}{\multirow{-2}{*}{\cellcolor{black!10}\raisebox{5mm}{\scalebox{1.3}{$\gamma \calH$}}}}  & \kappa & & \phantom{\ddots}\\
\hline
& \kappa & 0 & 1 & \phantom{\ddots}\\
& & 1 & 0 & \ddots \\
\phantom{\ddots} & \phantom{\ddots}& \phantom{\ddots} & \ddots & \ddots
\end{array}
\right)
\end{equation}
Here $\gamma$ is any deterministic constant, and $\kappa$ is either a random variable (independent of $\calH$) with a given distribution or deterministic $\kappa=1$.

Such an operator 
is natural from the point of view of physics: the random matrix part corresponds to a complex physical system of particles whose interactions are unknown, and the discrete Laplacian part corresponds to a lead attached via some coupling of strength $\kappa$.

We are interested in the spectral properties of the operator $\widetilde{\calH}$, namely in the locations of its eigenvalues and resonances (see Section~\ref{ss:finite range} below). In Theorem~\ref{th:resonances}, which is our main result, we compute the joint distribution of eigenvalues and resonances of $\widetilde{\calH}$ for the case of random $\kappa$. See remarks after the theorem for the case of deterministic $\kappa=1$.

The proof involves two main steps: first is to apply the Dumitriu--Edelman~\cite{DE,Trotter} tridiagonalization procedure to reduce $\widetilde{\calH}$ to a Jacobi operator on $\ell^2(\bbZ_+)$; second is to employ the (suitably modified) Geronimo--Case equations~\cite{GC} (see also Damanik--Simon~\cite[Appendix A]{DS2}) to access the Jost function whose zeros determine the locations of eigenvalues and resonances.

Aside from the physical importance, interest to the resonances comes from the fact that the locations of eigenvalues and resonances allow to fully recover the spectral measure of our operator $\widetilde{\calH}$. 
For the background on spectral theory of Jacobi operators we refer the reader to the monographs of Simon~\cite{Rice} and Teschl~\cite{Teschl}. 
The resonance problem for Jacobi operators was the topic of~\cite{BNW_Jac,DS2,Geronimo,GC,GI,Gol16,IK3,K_spectral,K_finite,MNSW,MW}, among many others. 
A closely related scattering theory for Jacobi operators is discussed in~\cite{Teschl}.

An operator-based approach to the asymptotics of the Dumitriu--Edelman Jacobi matrices was studied by Ram\'{i}rez--Rider--Vir\'{a}g~\cite{RRV} and Valk\'{o}--Vir\'{a}g~\cite{VV09}, see also subsequent papers by the same authors. There is also a vast literature on the Jacobi (or discrete Schr\"{o}dinger) operators with random coefficients, in particular in connection to the Anderson model, which we will not attempt to review here.

Random matrix approach to open quantum systems has two other alternatives to $\widetilde{\calH}$:  through non-Hermitian perturbations of Hermitian random matrices and through non-unitary perturbations of unitary random matrices -- see~\cite{FyoSav,FyoSom97,FyoSom,KhoSom} and references therein. Theory of orthogonal polynomials is applicable in both of these scenarios as well: see~\cite{KK_truncations,K_hermitian}.
%
%



%
%



The organization of the paper is as follows. In Sections~\ref{ss:finite range}--\ref{ss:Locations} we provide the background from the theory of Jacobi operators, including  properties  of resonances, the Jost function (perturbation determinant), and the Geronimo--Case equations. In Section~\ref{s:RMT} we state our main result and provide the proof.

\section{Jacobi operators}
\subsection{Finite range operators and perturbation determinants}\label{ss:finite range}
By a Jacobi operator we call a tridiagonal operator acting on $\ell^2(\bbZ_+)$ of the form
\begin{equation}\label{eq:jacobi}
\calJ(\mathbf{a},\mathbf{b})
=
\left(
\begin{array}{cccc}
b_1&a_1&0& \phantom{\ddots}\\
a_1&b_2&a_2& \ddots\\
0&a_2&b_3&\ddots \\
\phantom{\ddots} &\ddots&\ddots&\ddots
\end{array}\right),
\end{equation}
where $\mathbf{a} = \{a_j\}_{j=1}^\infty$, $\mathbf{b} = \{b_j\}_{j=1}^\infty$ have $a_j > 0$ and $b_j\in\bbR$.

The case $\mathbf{a} = \{1\}_{j=1}^\infty$, $\mathbf{b} = \{0\}_{j=1}^\infty$ corresponds to $\calJ_0$, the discrete Laplacian on $\bbZ_+$, and will be referred to as the free Jacobi operator.

For $s\ge 0$, we denote by $\calT^{[2s]}$ the set of all Jacobi operators that have $a_{j} = 1$, $b_{j}=0$ for $j > s$ and $a_{s} \ne 1$. We denote by $\calT^{[2s+1]}$ the set of all Jacobi operators that have $a_{j} = 1$, $b_{j}=0$ for $j > s+1$, and $a_{s} = 1$, but $b_s \ne 0$.
We denote $\calT^{[k\ge0]}$ to be the set of all Jacobi operators that are finite rank perturbations of the free one. It is the disjoint union of all $\calT^{[k]}$, $k\ge0$.

The spectral measure  $\mu$  (with respect to the vector $\mathbf{e}_1 = (1,0,0,\ldots)^T$) of any operator $\calJ\in \calT^{[k\ge0]}$ is of the form
\begin{equation}\label{eq:spectral}
d\mu(x) = \frac{\sqrt{4-x^2}}{a(x)} 1_{[-2,2]}(x) \,dx + d\mu_{p.p.}
\end{equation}
where  $a(x)$ is a polynomial and $\mu_{p.p.}$ contains finitely many pure points in $\bbR\setminus[-2,2]$ whose locations form a subset of the set of zeros of $a(x)$ (see ~\cite{Geronimo,DS2}).

The $m$-function of $\calJ$,
$$
m(z) = \int_\bbR \frac{d\mu(x)}{x-z},
$$
is meromorphic in $\bbC\setminus[-2,2]$ with poles at the pure points of $\mu$ (eigenvalues of $\calJ$). By ~\eqref{eq:spectral}, $m$ has a meromorphic continuation through $[-2,2]$ to a second copy of $\bbC\setminus[-2,2]$. Poles of $m$ on this second sheet are typically referred to as the resonances of $\calJ$. 

Let $\bbD=\{z:|z|<1\}$. For $z\in\bbC\setminus\overline{\bbD}$, we define
$$
M^*(z) = - m (z+z^{-1}).
$$
From the arguments in the previous paragraph, $M^*$ can be meromorphically continued from $\bbC\setminus\overline{\bbD}$ to $\bbC$. If $\calJ\in\calT^{[k]}$ then $M^*$ has precisely $k$ poles in $\bbC\setminus\{0\}$ counted with multiplicity.
Note that our $M^*(z)$ function is $M(1/z)$ in the notation of~\cite{DS2,Rice}.


For $\calJ\in\calT^{[k]}$ let us define the perturbation determinant
$$
L(z) = \det\left[(\calJ - z-z^{-1}))(\calJ_0 - z-z^{-1}))^{-1}\right], \qquad z\in\bbD.
$$
Then  $L(z)$ is a polynomial of degree $k$ with $L(0)=1$. Up to a normalization constant, $L(z)$ is equal to the Jost function of $\calJ$ (see~\cite{KS} for its properties). It will be convenient to work with the following polynomial instead:
$$
L^*(z) = z^k L(1/z).
$$
Then for $\calJ\in\calT^{[k]}$,  $L^*(z)$ is a \emph{monic} polynomial of degree $k$. It has
zeros at the poles of $M^*$ in $\bbC\setminus\{0\}$ counted with multiplicity. 

Zeros $z_j$ of $L^*$ in $\bbC\setminus\overline{\bbD}$ are in one-to-one correspondence with the eigenvalues $z_j+z_j^{-1}$ of $\calJ$, and zeros $z_j$ of $L^*$ in $\overline{\bbD}\setminus\{0\}$ are in one-to-one correspondence with the resonances $z_j+z_j^{-1}$ of $\calJ$ (counted with multiplicity). In order to simplify presentation, we will therefore refer to zeros of $L^*$ themselves as the eigenvalues, resp. resonances, of $\calJ$, with the  post-application of the Joukowsky map $z\mapsto z+z^{-1}$ being implicitly understood.


\subsection{Geronimo--Case equations}\label{ss:GC}
Let $\calJ\in\calT^{[k]}$, $\mu$ be its spectral measure, and $P_n(z)$ ($n\ge 0$) be the degree $n$ monic orthogonal polynomial associated with $\mu$. 
For each $j\ge 0$ we define
$$
K_{2j}(z) = K_{2j+1} (z)= z^j P_j(z+z^{-1}).
$$
Note that $K_{2j}=K_{2j+1}$ is a monic polynomial of degree $2j$.

For each $0\le j\le k$, let $\hat{\calJ}_j$ be the unique Jacobi operator that maximizes the number of zero entries in $\calJ-\hat{\calJ}_j$ under the restriction that $\hat{\calJ}_j\in \calT^{[j]}$. In particular, $\hat{\calJ}_0$ is the free Jacobi operator, and $\hat{\calJ}_k = \calJ$. Let $L_j^*(z)$ be the polynomial $L^*(z)$ (see the previous subsection) for the Jacobi operator $\hat{\calJ}_j$:
$$
L_j^*(z) := L^*(z; \hat{\calJ}_j).
$$
Recall that each $L_j$ is monic and of degree $j$.

Then the system of polynomials $\{K_j,L_j\}$ satisfies the recurrence relation below, which we call the Geronimo--Case equations. They have been modified compared with~\cite{GC,DS2}: e.g., in the notation of~\cite{DS2}, their $C_n(z)$ and $G_n(z)$ are ours $K_{2n}(z)$ and $z^{2n}L_{2n}(1/z)$, respectively. Taking this change into account, the Geronimo--Case equations~\cite[(A.19)]{DS2} take the form
\begin{equation}\label{eq:GC1}
\left(
\begin{array}{c}
L^*_{2k+2}(z) \\
K_{2k+2}(z)
\end{array}
\right)
=
\left(
\begin{array}{cc}
z & - (a_{k+1}^2-1) \\
z & 1
\end{array}
\right)
\left(
\begin{array}{c}
L^*_{2k+1}(z) \\
K_{2k+1}(z)
\end{array}
\right)
\end{equation}
and
\begin{equation}\label{eq:GC2}
\left(
\begin{array}{c}
L^*_{2k+1}(z) \\
K_{2k+1}(z)
\end{array}
\right)
=
\left(
\begin{array}{cc}
z & - b_{k+1} \\
0 & 1
\end{array}
\right)
\left(
\begin{array}{c}
L^*_{2k}(z) \\
K_{2k}(z)
\end{array}
\right)
\end{equation}
with the initial conditions $L^*_0(z)=K_0(z)=1$. 

In the next lemma we collect some of the properties of polynomials $L_j^*, K_j$ that we will need in Section~\ref{s:RMT} below.

\begin{lemma}\label{lm:jost}
For a given $m$, let
$L^*_m(z) = z^{m} + u_{m-1} z^{m-1} +  u_{m-2} z^{m-2} +\ldots +  u_{1} z + u_{0} = \prod_{j=1}^m (z-z_j)$.
Then
\begin{itemize}
\item[(i)]
$$
(-1)^m \prod_{j=1}^m z_j = u_0  =
\begin{cases}
 1-a^2_{m/2} & \mbox{if } m \bmod 2 =0\\
 -b_{(m+1)/2} & \mbox{if } m \bmod 2 =1
\end{cases}
$$
\item[(ii)]
$$
-\sum_{j=1}^{m} z_j = -\sum_{j=1}^{\lfloor \tfrac{m+1}{2} \rfloor} b_j =u_{m-1} .
$$
\item[(iii)]
$$
\sum_{\substack{j,k=1 \\ j<k}}^{m} z_j z_k = \sum_{\substack{j,k=1 \\ j<k}}^{\lfloor \tfrac{m+1}{2} \rfloor}  b_j b_k - \sum_{j=1}^{\lfloor \tfrac{m}{2} \rfloor}  (a_j^2-1)= u_{m-2}.
$$
\item[(iv)]
$$
\sum_{j=1}^{m} z_j^2 = \sum_{j=1}^{\lfloor \tfrac{m+1}{2} \rfloor} b^2_j +
2 \sum_{j=1}^{\lfloor \tfrac{m}{2} \rfloor}  (a_j^2 - 1) = u_{m-1}^2-2u_{m-2} .
$$
\item[(v)]
\begin{equation}\label{eq:result}
\prod_{\substack{j,k=1 \\ j<k}}^{m} (1-z_j \bar{z}_k) \prod_{j=1}^m \frac{1}{1-z_j^2} = \prod_{j=1}^{\lfloor \tfrac{m}{2} \rfloor} a_j^{4j}.
\end{equation}
\end{itemize}
\end{lemma}
\begin{proof}
(i) Note that $K_j(0)=1$, so this part follows immediately from~\eqref{eq:GC1} and~\eqref{eq:GC2} by plugging in $z=0$.

(ii) Since $K_{2k+1}$ is of degree $2k$,~\eqref{eq:GC1} shows that $z^{2k+1}$-coefficient of $L_{2k+2}^*$ is equal to the $z^{2k}$-coefficient of $L_{2k+1}^*$. ~\eqref{eq:GC2} shows that $z^{2k}$-coefficient of $L_{2k+1}^*$ is equal to the $z^{2k-1}$-coefficient of $L_{2k}^*$ minus $b_{k+1}$. An induction on $m$ then completes the proof.

(iii) can be shown in the exact same way by considering the terms one degree lower.

(iv) is immediate from (ii) and (iii) and $\sum z_j^2 = (\sum z_j)^2 - 2 \sum_{j<k} z_j z_k$.

(v)
For a polynomial $p$ with real coefficients of degree $j$ we define the operation $p^*(z):=z^j p(1/z)$. Then $(p^{*})^*= p$, so we define $L_j(z) = (L_j^*)^*$. Using $K_j^* = K_j$ and the recurrences~\eqref{eq:GC1} and~\eqref{eq:GC2}, we deduce
\begin{equation}\label{eq:geronimus}
K_{2k}(z)=K_{2k+1}(z)=\frac{L_{2k}(z)-z^2 L^*_{2k}(z)}{1-z^2} 
\end{equation}
and
\begin{equation}\label{eq:geronimus2}
K_{2k}(z)=K_{2k+1}(z)
 = \frac{L_{2k+1}(z)-z L^*_{2k+1}(z)}{1-z^2}.
\end{equation}


Let $\{z^{(k)}_j\}$ be the zeros of $L^*_k$ and let $\{\lambda_j^{(k)}\}$ be the zeros of $K_{k} $. Denote the left-hand side of~\eqref{eq:result} by $R_m$. Then for $m=2k$ even, we get
$$
R_{2k} = \prod_{j=1}^{2k} \frac{L_{2k}(z^{(2k)}_j)}{1-(z_j^{(2k)})^2} = \prod_{j=1}^{2k} K_{2k}(z_j^{(2k)})
$$
by~\eqref{eq:geronimus}. This can be further rewritten as
$$
R_{2k} = \prod_{j,s=1}^{2k} (z_j^{(2k)}-\lambda_s^{(2k)}) = \prod_{j=1}^{2k} L_{2k}^*(\lambda_j^{(2k)}) = \prod_{j=1}^{2k} \frac{1}{\lambda_j^{(2k)}} L^*_{2k+1}(\lambda_j^{(2k)}),
$$
where we used~\eqref{eq:GC2}. Note that $\prod_{j}^{(2k)} \lambda_j^{(2k)}$ is equal to the last coefficient of $K_{2j}$ which is $1$. So we get
$$
R_{2k} = \prod_{j=1}^{2k}  L^*_{2k+1}(\lambda_j^{(2k)}) = \prod_{j=1}^{2k+1}  K_{2k}(z_j^{(2k+1)})  = \prod_{j=1}^{2k+1}  \frac{L_{2k+1}(z_j^{(2k+1)})}{1-(z_j^{(2k+1)})^2}=R_{2k+1},
$$
where we used~\eqref{eq:geronimus2}.

For $m=2k+1$, following analogous steps, we get:
\begin{align*}
R_{2k+1} & =   \prod_{j=1}^{2k+1} \frac{L_{2k+1}(z^{(2k+1)}_j)}{1-(z_j^{(2k+1)})^2} = \prod_{j=1}^{2k+1} K_{2k+1}(z_j^{(2k+1)})
= \prod_{j=1}^{2k+1}\prod_{s=1}^{2k} (z_j^{(2k+1)}-\lambda_s^{(2k+1)})
\\
&  = \prod_{j=1}^{2k} L_{2k+1}^*(\lambda_j^{(2k+1)}) = \prod_{j=1}^{2k} \frac{1}{\lambda_j^{(2k+1)}} L^*_{2k+2}(\lambda_j^{(2k+1)})
= \prod_{j=1}^{2k+2}  K_{2k+1}(z_j^{(2k+2)})
\\
&   = \prod_{j=1}^{2k+2}  \frac{1}{a_{k+1}^2} K_{2k+2}(z_j^{(2k+2)})
 =  \frac{1}{a_{k+1}^{4(k+1)}}  \prod_{j=1}^{2k+2}  \frac{L_{2k+2}(z_j^{(2k+2)})}{1-(z_j^{(2k+2)})^2}= \frac{1}{a_{k+1}^{4(k+1)}} R_{2k+2},
\end{align*}
where in the last line we used $a_{k+1}^2 K_{2k+1} = K_{2k+2}-L^*_{2k+2}$ which is a consequence of~\eqref{eq:GC1}. Combining our recurrences for $R_j$'s, we obtain~\eqref{eq:result}.
\end{proof}

\subsection{Locations of resonances and eigenvalues}\label{ss:Locations}

It was shown in~\cite{DS2} (see also~\cite{K_finite}) that the set of resonances and eigenvalues of $\calJ\in \calT^{[k\ge0]}$ uniquely determines $\calJ$. In fact the following sets $S(k)$ classify all possible configurations of resonances and eigenvalues of $\calJ\in\calT^{[k]}$, $k\ge 0$. 

\begin{definition}
Denote by $S(k)$ the set of all possible $\{z_j\}_{j=1}^k$ in $(\bbC\setminus\{0\})^k$ that satisfy the following conditions:

\begin{itemize}
\item[(i)] $z_j$'s are real or come in complex-conjugate pairs.
\item[(ii)] $z_j$'s that lie in $\bbC\setminus\bbD$ are real and of multiplicity $1$.
\item[(iii)] Let $1<x_1<x_2<\ldots$ be the positive $z_j$'s on $\bbC\setminus\overline{\bbD}$ counted with multiplicity. Then
\begin{itemize}
\item[(a)] There is an even number of $z_j$'s (counted with multiplicity) on $(x_1^{-1},1]$;
\item[(b)] There is an odd number of $z_j$'s (counted with multiplicity) on $(x_{m+1}^{-1},x_m^{-1})$ $(m\ge1)$;
\item[(c)] None of $z_j$'s is equal to $x_m^{-1}$ $(m\ge1)$;
\end{itemize}
\item[(iv)] Let $\ldots<y_2<y_1<-1$ be the negative $z_j$'s on $\bbC\setminus\overline{\bbD}$  counted with multiplicity. Then
\begin{itemize}
\item[(a)] There is an even number of $z_j$'s (counted with multiplicity) on $[-1,y_1^{-1})$;
\item[(b)] There is an odd number of $z_j$'s (counted with multiplicity) on $(y_m^{-1},y_{m+1}^{-1})$ $(m\ge1)$;
\item[(c)]  None of $z_j$'s is equal to $y_m^{-1}$ $(m\ge1)$.
\end{itemize}
\end{itemize}
\end{definition}


\section{Random matrices coupled to the Laplacian}\label{s:RMT}


Let $N (0,1)$ 
be the real 
normal random variable with mean $0$ and variance 1. 
Let $Y$ be an $n\times n$ matrix with independent identically distributed real entries chosen from $N(0,1)$.
Then we say that the random matrix $X=\tfrac{1}{2}(Y + Y^*) \tfrac{\sqrt{2}}{\sqrt{\beta n}}$ (where $\beta=1$) belongs to the \emph{Gaussian orthogonal ensemble}. 

Similarly, let $Y$ be an $n\times n$ matrix with independent identically distributed complex entries chosen from $N(0,1)+ N(0,1) \rm{i}$.
Then we say that the random matrix $X=\tfrac{1}{2}(Y + Y^*) \tfrac{\sqrt{2}}{\sqrt{\beta n}}$ (where $\beta=2$) belongs to the \emph{Gaussian unitary ensemble}.

Finally, let $Y$ be an $n\times n$ matrix with independent identically distributed quternionic entries chosen from $N(0,1)+ N(0,1) \rm{i}+N(0,1) \rm{j}+N(0,1) \rm{k}$.
Then we say that the random matrix $X=\tfrac{1}{2}(Y + Y^R) \tfrac{\sqrt{2}}{\sqrt{\beta n}}$ (where $\beta=4$) belongs to the \emph{Gaussian symplectic ensemble}.

We denote these ensembles by $GOE_n$, $GUE_n$, $GSE_n$, respectively.

Note that we chose the extra scaling factor $\tfrac{\sqrt{2}}{\sqrt{\beta n}}$. This is chosen so that  the empirical density of states of each of these ensembles converges to semicircle distribution $\tfrac{1}{2\pi}\sqrt{4-x^2}\,dx$ on $[-2,2]$. With such normalization, the joint eigenvalue density of $GOE_n$, $GUE_n$, $GSE_n$ is proportional to
\begin{equation}\label{eq:gaussian}
\prod_{\substack{j,k=1 \\ j<k}}^{n} |\lambda_j - \lambda_k|^\beta \prod_{j=1}^n e^{-\frac{\beta n }{4} \lambda_j^2 } d\lambda_j,
\end{equation}
($\beta=1,2,4$, respectively).



Now let us state the main result of the paper. Recall that by ``eigenvalues'' and ``resonances'' we call  the zeros of the polynomial $L^*$ in $\bbC\setminus\overline{\bbD}$ and $\overline{\bbD}\setminus\{0\}$, respectively, see the discussion in Section~\ref{ss:finite range}.
\begin{theorem}\label{th:resonances}
Let $\widetilde{\calH}$ be given by~\eqref{eq:main_operator} where $\calH$ is from $GOE_n$, $GUE_n$ or $GSE_n$; $\gamma\ne0$ is a given constant; and $\kappa$ is a random variable distributed on $(0,\infty)$ according to $F(\kappa) d\kappa$, independently of $\calH$. Then resonances and eigenvalues $\{z_j\}_{j=1}^{2n}$ of $\widetilde{\calH}$ are jointly distributed on the set $S(2k)$ according to
\begin{multline}\label{eq:resonance_distr}
\frac{1}{d_{2n,\beta}} \prod_{\substack{j,k=1 \\ j<k}}^{2n}  |z_j-z_k| \prod_{\substack{j,k=1 \\ j<k}}^{2n} |1-z_j \bar{z}_k|^{\frac{\beta-2}{2}} \prod_{j=1}^{2n} e^{-\frac{\beta n }{4\gamma^2}z_j^2} \Big|\frac{1-|z_j|^2}{1-z_j^2}\Big|^{\frac{\beta-2}{4}} \\
\times e^{\frac{\beta n \kappa^2}{2\gamma^2}} \frac{F(\kappa)}{\kappa^{\beta n -1}} \,
\Big| \bigwedge_{j=1}^{2n} dz_j \Big|,
\end{multline}
where $\kappa=\sqrt{1-\prod_{j=1}^{2n} z_j}$ and
$$
d_{2n,\beta} =
\pi^{n/2} 2^{n/2+1} e^{\tfrac{\beta n^2}{2\gamma^2}} \left(\frac{2\gamma^2}{\beta n }\right)^{\tfrac{n}{2}+\tfrac{\beta n (n-1)}{4}} \prod_{j=1}^{n-1} \Gamma(\tfrac{\beta j}{2}).
$$
\end{theorem}
\begin{remarks}
1. The wedge notation we use above is defined as follows. Let $\{z_j\}_{j=1}^m$ (in the theorem above, $m=2n$) be a random point process that consists of $M$ complex-conjugate (non-real) points  and $L$ real points. $M$ and $L$ are random but satisfy $0\le M \le \lfloor \tfrac{m}{2} \rfloor$, $0\le L \le m$, $L+2M = m$. Then for functions $f:\bbC^m\to\bbC$ invariant under permutation of its variables, we define
\begin{multline}\label{eq:wedges}
\int_X f(z_1,\ldots,z_{m}) \, \Big| \bigwedge_{j=1}^{m} dz_j \Big| \\ := \sum_{M=0}^{\lfloor \frac{m}{2} \rfloor} 2^M \frac{1}{M!L!2^M} \int_{X\cap X_{L,M}} f(x_1+ iy_{1}, x_1-i y_1,\ldots,x_{M}+ iy_{M}, x_M-iy_M,r_1,\ldots,r_{L}) \\
\times dx_1 dy_1 \ldots dx_M dy_M \, dr_{1}\ldots dr_{L},
\end{multline}
where
\begin{multline*}
X_{L,M} = \left\{(x_1+ iy_{1}, x_1-i y_1,\ldots,x_{M}+ iy_{M}, x_M-iy_M,r_{1},\ldots,r_{L})\in X: \right. \\
\left.x_j+iy_j\in\bbC\setminus\bbR \mbox{ for }1\le j \le M; r_j\in\bbR \mbox{ for }1\le j \le L\right\}.
\end{multline*}
Note that $2^M$ here comes from $|d(x+iy)\wedge d(x-iy)| = 2\,dx\,dy$ and $M!L!2^M$ comes from counting vectors in  $X_{L,M}$ that represent the same configuration $\{z_j\}_{j=1}^m$. See~\cite[Sect. 2--3]{BorSin} for a more careful and rigorous discussion of these types of measures.

2. When $\kappa$ is deterministic and equal to $1$, then there are $2n-1$ resonances/eigenvalues. They belong to $S(2n-1)$ and following along the same lines as in the proof below, one can show that their joint distribution is
\begin{equation}\label{eq:resonance_distr2}
\frac{1}{d_{2n-1,\beta}} \prod_{\substack{j,k=1 \\ j<k}}^{2n-1}  |z_j-z_k| \prod_{\substack{j,k=1 \\ j<k}}^{2n-1} |1-z_j \bar{z}_k|^{\frac{\beta-2}{2}} \prod_{j=1}^{2n-1} e^{-\frac{\beta n }{4\gamma^2}z_j^2} \Big|\frac{1-|z_j|^2}{1-z_j^2}\Big|^{\frac{\beta-2}{4}}
\,
\Big| \bigwedge_{j=1}^{2n} dz_j \Big|,
\end{equation}
where
$$
d_{2n-1,\beta} =
\pi^{n/2} 2^{n/2} e^{\tfrac{\beta n(n-1)}{2\gamma^2}} \left(\frac{2\gamma^2}{\beta n }\right)^{\tfrac{n}{2}+\tfrac{\beta n (n-1)}{4}} \prod_{j=1}^{n-1} \Gamma(\tfrac{\beta j}{2}).
$$

3. One can also work out the case when $\kappa$ is deterministic but not $1$. Note that in that case the eigenvalues/resonances belong to the subset of $S(2n)$ given by $\prod_{j=1}^{2n} z_j = 1-\kappa^2$ (see Lemma~\ref{lm:jost}(i)). See~\cite{K_hermitian} for an analogue of this for non-Hermitian perturbations of finite matrices.

4. Compare~\eqref{eq:resonance_distr}/\eqref{eq:resonance_distr2} with~\eqref{eq:gaussian}. Note from the recurrence~\eqref{eq:GC1},~\eqref{eq:GC2} that $L_j^*$ and $K_j$ can be viewed as perturbations of each other. Their zeros are precisely resonances/eigenvalues of $\widetilde{\calH}$ and eigenvalues of $\calH$ (after the inverse of $z+z^{-1}$ map), respectively.
\end{remarks}
\begin{proof}
Every $n\times n$ matrix can be tridiagonalized via the repeated application of the Householder transformations. Applying this to a random matrix $\calH_n$ taken from one of the $GOE_n$, $GUE_n$, $GSE_n$ ensembles, Dumitriu--Edelman showed
that there exists a unitary matrix $U_n$ such that 
\begin{equation}\label{eq:finite_jacobi}
\calJ_n = U_n^* \mathcal{H}_n U_n = \left(
\begin{array}{ccccc}
s_1&t_1&0& &\\
t_1&s_2&t_2&\ddots &\\
0&t_2&s_3&\ddots & 0 \\
 &\ddots&\ddots&\ddots & t_{n-1} \\
 & & 0 & t_{n-1} & s_n
\end{array}\right).
\end{equation}
Moreover, $U_n$ is independent of $\calJ_n$, satisfies
\begin{equation}\label{eq:fixE1}
U_n \mathbf{e}_1 = U_n^* \mathbf{e}_1 =  \mathbf{e}_1,
\end{equation}
and the joint distribution of the coefficients $\{s_j\}_{j=1}^n$  and $\{t_j\}_{j=1}^{n-1}$ is
\begin{equation}\label{eq:jacobi_distr}
\frac{1}{c_{n,\beta}} \prod_{j=1}^{n-1} t_j^{\beta (n-j) - 1} e^{-\beta n t_j^2/2} \, dt_j \prod_{j=1}^n e^{- \beta n s_j^2/4} \, ds_j,
\end{equation}
where
\begin{equation}\label{eq:jacobi_norm}
c_{n,\beta}  = \frac{\pi^{n/2}}{2^{n/2-1}} \left(\frac{2}{\beta n }\right)^{\tfrac{n}{2}+\tfrac{\beta n (n-1)}{4}} \prod_{j=1}^{n-1} \Gamma(\tfrac{\beta j}{2})
\end{equation}
(this follows from Dumitriu--Edelmam~\cite{DE} after rescaling).
Here $\beta=1,2,4$ for $GOE_n$, $GUE_n$, $GSE_n$, respectively. In fact, for any $0<\beta<\infty$, $\calJ_n$ in~\eqref{eq:finite_jacobi} with~\eqref{eq:jacobi_distr},~\eqref{eq:jacobi_norm} is a well-defined random matrix, whose eigenvalue distribution is (proportional to)~\eqref{eq:gaussian}.

Now let $\widetilde{\calH}$ be given by~\eqref{eq:main_operator} where $\calH=\calH_n$ is from $GOE_n$, $GUE_n$ or $GSE_n$; $\gamma\ne0$ is a given constant; and $\kappa$ is a random variable distributed on $(0,\infty)$ and independent from $\calH$. Let $R$ be an $n\times n$ matrix with $1$'s on the anti-diagonal and $0$'s everywhere else. By the invariance of the Gaussian ensembles, $R^* \calH_n R^*$ belongs to the same random matrix ensemble as $\calH_n$. Now define $U_n$ as above but applied to random matrix $R^* \calH_n R$ instead of $\calH_n$. Then $U_n^* R^* \calH_n R U_n = \calJ_n$ and ~\eqref{eq:fixE1} holds. Define $\widetilde{U} = (R^* U_n R) \oplus I$ on $\ell^2(\bbZ_+)$. Then~\eqref{eq:fixE1} implies $RU_n R^* \mathbf{e}_n = RU_n^* R^* \mathbf{e}_n=\mathbf{e}_n$, so that we get
$$
\widetilde{U}^* \widetilde{\calH} \widetilde{U} =
\left(
\begin{array}
{  c r | c  c  c }
\multicolumn{2}{c|}{{\cellcolor{black!10}}} & & & \phantom{\ddots}\\
\multicolumn{2}{c|}{\multirow{-2}{*}{\cellcolor{black!10}\raisebox{5mm}{\scalebox{1.3}{$\gamma R\calJ_n R^*$}}}}  & \kappa & & \phantom{\ddots}\\
\hline
& \kappa & 0 & 1 & \phantom{\ddots}\\
& & 1 & 0 & \ddots \\
\phantom{\ddots} & \phantom{\ddots}& \phantom{\ddots} & \ddots & \ddots
\end{array}
\right)
$$
This means that $\widetilde{H}$ is unitarily equivalent to a Jacobi operator $\calJ(\mathbf{a},\mathbf{b})$ (see~\eqref{eq:jacobi}) with $\mathbf{a}=\{\gamma t_{n-1},\gamma t_{n-2},\ldots,\gamma t_1, \kappa, 1,1,\ldots\}$ and $\mathbf{b}=\{\gamma s_n,\gamma s_{n-1},\ldots,\gamma s_1, 0,0,\ldots\}$ (in other words, the Dumitriu--Edelman coefficients are order-reversed, scaled by $\gamma$ and then coupled to the free Jacobi operator with coupling $\kappa$).


As stated in Section~\ref{ss:Locations}, there is a one-to-one correspondence between the $2n$ Jacobi coefficients $\mathbf{a},\mathbf{b}\in \bbR_+^n\times \bbR^n$ and $2n$ zeros  in $S(2n)$  of the (reversed) perturbation determinant $L_{2n}^*$.
We will compute the Jacobian of this transformation by computing one step at a time:

\begin{lemma}
Let
$$
L^*_j(z) = z^{j} + u^{(j)}_{j-1} z^{j-1} +  u^{(j)}_{j-2} z^{j-2} +\ldots +  u^{(j)}_{1} z + u^{(j)}_{0}.
$$
Then~\eqref{eq:GC1} and~\eqref{eq:GC2} imply
\begin{equation}\label{eq:sbs1}
\det\frac{\partial\left(u_{2k}^{(2k+1)},u_{2k-1}^{(2k+1)},\ldots , u_{0}^{(2k+1)}\right)}{\partial\left(u_{2k-1}^{(2k)},u_{2k-2}^{(2k)},\ldots , u_{0}^{(2k)},b_{k+1}\right)} = -1
\end{equation}
and
\begin{equation}\label{eq:sbs2}
\det\frac{\partial\left(u_{2k+1}^{(2k+2)},u_{2k}^{(2k+2)},\ldots , u_{0}^{(2k+2)}\right)}{\partial\left(u_{2k}^{(2k+1)},u_{2k-1}^{(2k+1)},\ldots , u_{0}^{(2k+1)},a_{k+1}\right)} = -2 a_{k+1}^{2k+1}.
\end{equation}
\end{lemma}
\begin{proof}
Let $K_{2k}(z)=K_{2k+1}(z) = z^{2k}+c_{2k-1} z^{2k-1} + \ldots +c_1 z + 1$. We will also put $c_{2k}=c_0 = 1$ and $c_j=0$ for $j<0$ or $j>2k$. Note that $c_j = c_{2k-j}$ for all $j$ since $K_{2k} = K_{2k}^*$.

Equality~\eqref{eq:geronimus} implies that for each $j \le k$, $c_j-c_{j-2} = u^{(2k)}_{2k-j}-u^{(2k)}_{j-2}$, which shows that $c_j$ (for $j\le k$)
does not depends on the coefficients $u^{(2k)}_l$ with $j-1 \le l \le 2k-j-1$.

Now,~\eqref{eq:GC2} implies $u^{(2k+1)}_j = u^{(2k)}_{j-1} - b_{k+1} c_j$.
Using this, we can show that the Jacobian matrix
\begin{equation}\label{eq:det_interm}
\frac{\partial\left(u_{0}^{(2k+1)},u_{2k}^{(2k+1)},u_{1}^{(2k+1)},u_{2k-1}^{(2k+1)},\ldots, u_{k-1}^{(2k+1)},u_{k+1}^{(2k+1)},u_{k}^{(2k+1)}\right)}{\partial\left(b_{k+1}, u_{2k-1}^{(2k)},u_{0}^{(2k)},u_{2k-2}^{(2k)},\ldots, u_{k-2}^{(2k)},u_{k}^{(2k)},u_{k-1}^{(2k)}\right)}
\end{equation}
has a triangular structure. Indeed, $u^{(2k+1)}_0 =  -b_{k+1}$, $u^{(2k+1)}_{2k} = u^{(2k)}_{2k-1} - b_{k+1}$. Furthermore, $u^{(2k+1)}_1 = u^{(2k)}_{0} - b_{k+1} c_1$, $u^{(2k+1)}_{2k-1} = u^{(2k)}_{2k-2} - b_{k+1} c_1$; and as we saw earlier $c_1$ is a independent of $u^{(2k)}_l$ with $0 \le l \le 2k-2$. This can be continued by induction. The determinant of the triangular matrix~\eqref{eq:det_interm} is equal to the product of the diagonal entries, which equals to $-1$. This proves~\eqref{eq:sbs1}.

Similar arguments prove~\eqref{eq:sbs2}, with just one extra wrinkle. Equality~\eqref{eq:geronimus2} shows that for each $j \le k$, $c_j-c_{j-2} = u^{(2k+1)}_{2k-j+1}-u^{(2k+1)}_{j-1}$, which shows that $c_j$ (for $j\le k$) is equal to $-u^{(2k+1)}_{j-1}+d_j$, where $d_j$ does not depends on the coefficients $u^{(2k+1)}_l$ with $j-1\le l \le 2k-j$. Then we show that the Jacobian matrix
\begin{equation}\label{eq:det_interm2}
\frac{\partial\left(u_{2k+1}^{(2k+2)},u_{0}^{(2k+2)},u_{2k}^{(2k+2)},\ldots, u_{k-1}^{(2k+2)},u_{k+1}^{(2k+2)},u_{k}^{(2k+2)}\right)}{\partial\left(u_{2k}^{(2k+1)},a_{k+1}, u_{2k-1}^{(2k+1)},u_{0}^{(2k+1)},\ldots, u^{(2k+1)}_{k-2},u^{(2k+1)}_{k},u^{(2k+1)}_{k-1}\right)}
\end{equation}
has a triangular structure. Indeed, using~\eqref{eq:GC1}, we get $u^{(2k+2)}_j = u^{(2k+1)}_{j-1} + (1-a^2_{k+1}) c_j$. This implies $u_{2k+1}^{(2k+2)}=u_{2k}^{(2k+1)}$; $u^{(2k+2)}_0 =  1-a^2_{k+1}$, $u^{(2k+2)}_{2k} = u^{(2k+1)}_{2k-1} +(1- a^2_{k+1})$. Furthermore, $u^{(2k+2)}_1 = u^{(2k+1)}_{0} +(1-a_{k+1}^2) c_1 = a_{k+1}^2 u^{(2k+1)}_{0} +(1-a_{k+1}^2) d_1$; $u^{(2k+2)}_{2k-1} = u^{(2k+1)}_{2k-2} +(1- a^2_{k+1}) c_1$; and as we saw earlier $d_1$ depends only on $u_{2k}^{(2k+1)}$, while $c_1$ depends only on $u_{2k}^{(2k+1)}$ and $u_{0}^{(2k+1)}$. This together with an induction shows the triangular structure. The determinant of~\eqref{eq:det_interm2} is then equal to the product of the diagonal entries, which equals to $1\times(-2 a_{k+1})\times 1\times(a_{k+1}^2)^k$. This proves~\eqref{eq:sbs2}.
\end{proof}

Now we are ready to compute the main Jacobian.
\begin{lemma}\label{lm:Jacobian}
Let $\{a_j,b_j\}_{j=1}^n\in \bbR_+^n\times \bbR^n$ be the first Jacobi coefficients of $\calJ$, and let $\{z_j\}_{j=1}^{2n}\in S(2n)$ be the zeros of $L^*_{2n}$. Then the following change of variables holds true:
\begin{equation}\label{eq:jacobian}
\prod_{j=1}^n da_j \, db_j=
\frac{\prod_{j<k} |z_j-z_k|}{2^n \prod_{j=1}^{n} a_{j}^{2j-1}} \, \Big| \bigwedge_{j=1}^{2n} dz_j \Big|.
\end{equation}
\end{lemma}
\begin{proof}
Applying the previous lemma recursively, we obtain
$$
\det\frac{\partial\left(u_{2n-1}^{(2n)},u_{2n-2}^{(2n)},\ldots u_{0}^{(2n)}\right)}{\partial\left(b_1,a_1\ldots,b_{k-1},a_{k-1},b_k,a_k\right)} = 2^n \prod_{j=1}^{n} a_{j}^{2j-1}.
$$
Finally, the change of variables
$$
\prod_{j=0}^{2n-1} du_{j}^{(2n)} = \prod_{j<k} |z_j-z_k| \, \Big| \bigwedge_{j=1}^{2n} dz_j \Big|
$$
follows from the arguments in~\cite[Lemma 6.5]{KK_truncations} (we warn the reader of the missing factor $\tfrac{1}{M! L! 2^M}$ that is needed in~\cite[eq.(3.3)]{KK_truncations}). Combining the last two formulas, we obtain~\eqref{eq:jacobian}.
\end{proof}

Now recall that we are computing zeros of $L_{2n}^*$ for the Jacobi matrix $\calJ(\mathbf{a},\mathbf{b})$ with $\mathbf{a}=\{\gamma t_{n-1},\gamma t_{n-2},\ldots,\gamma t_1, \kappa, 1,1,\ldots\}$ and $\mathbf{b}=\{\gamma s_n,\gamma s_{n-1},\ldots,\gamma s_1, 0,0,\ldots\}$, where the distribution of $\{t_j,s_j\}$ is given in~\eqref{eq:jacobi_distr}. Performing the order-reversal and scaling, we obtain that the joint distribution of $\{a_j,b_j\}_{j=1}^n$ is
$$
\frac{1}{\widetilde{c}_{n,\beta}} \prod_{j=1}^{n-1} a_j^{\beta j - 1} e^{-\beta n a_j^2/(2\gamma^2)} \, da_j \prod_{j=1}^n e^{- \beta n b_j^2/(4\gamma^2)} \, db_j F(a_n)\,da_n ,
$$
where
$$
\widetilde{c}_{n,\beta}= \frac{\pi^{n/2}}{2^{n/2-1}} \left(\frac{2\gamma^2}{\beta n }\right)^{\tfrac{n}{2}+\tfrac{\beta n (n-1)}{4}} \prod_{j=1}^{n-1} \Gamma(\tfrac{\beta j}{2}).
$$
Applying Lemma~\ref{lm:Jacobian}, we obtain that this is equal to
$$
\frac{1}{2^n \widetilde{c}_{n,\beta}} \prod_{j=1}^{n-1} a_j^{(\beta-2) j } e^{-\beta n a_j^2/(2\gamma^2)} \prod_{j=1}^n e^{- \beta n b_j^2/(4\gamma^2)} \frac{F(a_n)}{a_n^{2n-1}} \prod_{j<k} |z_j-z_k| \, \Big| \bigwedge_{j=1}^{2n} dz_j \Big|.
$$
Now applying parts (i), (iv), and (v) of Lemma~\ref{lm:jost} easily leads to the distribution~\eqref{eq:resonance_distr}.
\end{proof}

\bigskip

\bibliographystyle{plain}
\bibliography{../../mybib,../../mybib_rmt}

\end{document}